\documentclass[11pt]{amsproc}
\usepackage{mathrsfs}
\usepackage{stmaryrd}
\usepackage{cases}
\usepackage{amssymb}
\usepackage{amsmath}
\usepackage{amsfonts}
\usepackage{graphicx}
\usepackage{amsmath,amstext,amsbsy,amssymb}
\usepackage{verbatim}
\newtheorem{theorem}{Theorem}[section]
\newtheorem{lemma}[theorem]{Lemma}

\newtheorem{corollary}[theorem]{Corollary}

\theoremstyle{definition}

\theoremstyle{remark}

\numberwithin{equation}{section} \errorcontextlines=0

\usepackage{hyperref}
\usepackage{graphicx}
\usepackage{epsfig}
\usepackage{amsmath}
\usepackage{epstopdf}

\begin{document}

\title{On genuine entanglement for tripartite systems}

\author{Hui Zhao}
\address{School of Mathematics, Faculty of Science, Beijing University of Technology, Beijing 100124, China}
\email{zhaohui@bjut.edu.cn}
\author{Lin Liu}
\address{School of Mathematics, Faculty of Science, Beijing University of Technology, Beijing 100124, China}
\email{17865515630@163.com}
\author{Zhi-Xi Wang}
\address{School of Mathematical Sciences, Capital Normal University, Beijing 100048, China}
\email{wangzhx@cnu.edu.cn}
\author{Naihuan Jing}
\address{Department of Mathematics, North Carolina State University, Raleigh, NC 27695, USA}
\email{jing@ncsu.edu}
\author{Jing Li}
\address{Interdisciplinary Research Institute, Faculty of Science, Beijing University of Technology, Beijing
100124, China}
\email{lijing@bjut.edu.cn}

\maketitle

\begin{abstract}
We investigate the genuine entanglement in tripartite systems based on partial transposition and the norm of correlation tensors of the density matrices. We first derive an analytical sufficient criterion to detect genuine entanglement of tripartite qubit quantum states combining with the partial transposition of the density matrices. Then we use the norm of correlation tensors to study genuine entanglement for tripartite qudit quantum states and obtain a genuine entanglement criterion by constructing certain matrices. With detailed examples our results are seen to be able to detect more genuine tripartite entangled states than previous studies.
\end{abstract}




\section{Introduction}	

Quantum entanglement is an important resource for quantum information processing with wide applications in quantum computing, secure communication and channel protocols.$^{1-4}$ Genuine multipartite entanglement (GME) is one of the important entanglement types that offers significant advantage in quantum tasks compared with bipartite ones.$^{5}$

It is known that any multipartite pure state can be written as a tensor product $|\varphi\rangle\langle\varphi|=|\varphi_{A}\rangle\langle\varphi_{A}\rangle\otimes|\varphi_{\bar{A}}\rangle\langle\varphi_{\bar{A}}|$ with respect to some bipartition $A\bar{A}$ ($A$ denoting a subset of subsystems and $\bar{A}$ its complement) is called biseparable. Any mixed state that can be decomposed into a convex sum of biseparable pure states is called biseparable. Consequently, any non-biseparable state is called genuinely multipartite entangled (GME). Let $H_i^d$, $i=1,2,3$, denote $d$-dimensional Hilbert spaces. A tripartite state $\rho$ over $H_{1}^{d}\otimes H_{2}^{d}\otimes H_{3}^{d}$ can be expressed as $\rho = \sum {{p_\alpha }} \left| {{\psi _\alpha }} \right\rangle \left\langle{{\psi_\alpha }} \right|$, where $0<p_\alpha\leq 1$, $\sum {{p_\alpha }}  = 1$, $\left| {{\psi _\alpha }} \right\rangle \in H_{1}^{d} \otimes H_{2}^{d}\otimes H_{3}^{d}$ are normalized pure states. If all $|\psi_{\alpha}\rangle$ are biseparable, namely, either $|\psi_{\alpha}\rangle=|\varphi_{\alpha}^{1}\rangle\otimes|\varphi_{\alpha}^{23}\rangle$ or $|\psi_{\beta}\rangle=|\varphi_{\beta}^{2}\rangle\otimes|\varphi_{\beta}^{13}\rangle$ or $|\psi_{\gamma}\rangle=|\varphi_{\gamma}^{3}\rangle\otimes|\varphi_{\gamma}^{12}\rangle$, where $|\varphi^1_\alpha\rangle$, $|\varphi^2_\beta\rangle$ and $|\varphi^3_\gamma\rangle$ denote pure states in $H_{1}^{d}$, $H_{2}^{d}$ and $H_{3}^{d}$, $|\varphi_{\alpha}^{23}\rangle$, $|\varphi_{\beta}^{13}\rangle$ and $|\varphi_{\gamma}^{12}\rangle$ denote pure states in $H_{2}^{d}\otimes H_{3}^{d}$, $H_{1}^{d}\otimes H_{3}^{d}$ and $H_{1}^{d}\otimes H_{2}^{d}$ respectively, then $\rho$ is said to be bipartite separable. Otherwise, $\rho$ is called genuine tripartite entangled. Therefore, a multipartite quantum state is not separable with respect to any bipartition will be called GME.$^{6}$

GME was studied in terms of linear and non-linear entanglement wit\-nesses.$^{7-8}$ Measurement of GME
based on concurrence and its lower bound was studied.$^{9}$
GME criteria for states of generic dimension using lower bounds of concurrence were obtained.$^{10}$ Sufficient conditions for detecting genuine tripartite entanglement and lower bounds of concurrence were presented.$^{11}$ Li et. al. investigated the genuine multipartite entanglement in terms of the norm of correlation tensors and multipartite concurrence.$^{12}$ The correlation tensors of quantum states were used to
 give a necessary condition for separability of multipartite quantum states.$^{13}$ Vicente and Huber have developed a general framework to
 use the same piece of information to detect both entanglement and GME for multipartite states of arbitrary dimension by using correlation tensors.$^{14}$ Some necessary conditions of separability and its generalized form in terms of the norms of correlation tensors out of
 density matrices were presented.$^{15}$ Also through correlation tensors, successful detection of multipartite entanglement for higher number of systems was presented.$^{16-17}$ The relation between the norm of correlation tensors and concurrence for tripartite qudit quantum systems have been established,$^{18}$ wherein an effective lower bound for concurrence and GME criteria were presented. The norm of the Bloch vectors for any quantum state with the number of subsystems less than or equal to four has been studied.$^{19}$ Despite of all these, detecting GME for more situations remains to be a priority for quantum computation due to its complexity and importance.

In this paper we present criteria of GME for tripartite quantum states based on the partial transpose and the norm of
correlation tensors of quantum states.$^{20}$ We construct certain matrices out of the density matrix in terms of bipartition
to give lower bounds for their norms, which enable us to derive new criteria for GME. We have compared our results with some
commonly used GME criteria and find that our results are more effective in several cases.

This paper is organized as follows. In Section 2, we obtain an analytical sufficient criterion for detecting the genuine entanglement of tripartite qubit quantum states by using partial transposition of the density matrix. By a detailed example, our results are seen to outperform some previously available results. In Section 3, we study genuine entanglement for tripartite qudit quantum systems and present a new sufficient condition of GME for multipartite quantum states based on the norm of correlation tensors by constructing certain matrices. We also give an example to show that our criterion detects more GME than some commonly used results. Comments and conclusions are given in Section 4.

\section{Detection of GME for tripartite quantum states}

We first consider the GME for tripartite qubit quantum states. Let $H_{j}^{d}$ $(j=1,2,3)$ be $d$-dimensional Hilbert spaces.
For our purpose, a density matrix is expanded in terms of the identity operator $I_{d}\in \mathbb C^{d\times d}$ and the traceless Hermitian generators $\lambda_{i_{j}}^{(j)}$ $(j=1,2,3$; $i_{j}=1,2,\cdots,d^{2}-1)$ of the Lie algebra $\mathfrak{su}(d)$. The generators can be easily constructed from any orthonormal basis $\{|a\rangle\}_{a=0}^{d-1}$ in $H_j^d$. 
Let $l,j,k$ be indices such that $0\leq l\leq d-2$ and $0\leq j<k\leq d-1$. Then, when $i_{j}=1,\cdots,d-1$, $\lambda_{i_{j}}^{(j)}=\sqrt{\frac{2}{(l+1)(l+2)}}(\sum\limits_{a=0}^{l}|a\rangle\langle a|-(l+1)|l+1\rangle\langle l+1|)$, when $i_{j}= d,\cdots,(d + 2)(d- 1)/2$, $\lambda_{i_{j}}^{(j)}=|j\rangle\langle k|+|k\rangle\langle j|$, and when $i_{j}= d(d + 1)/2,\cdots,d^{2}-1$, $\lambda_{i_{j}}^{(j)}=-i(|j\rangle\langle k|-|k\rangle\langle j|)$. Then a tripartite qubit quantum state $\rho$ on $H_{1}^{2}\otimes H_{2}^{2}\otimes H_{3}^{2}$ can be expressed in the following form:
\begin{eqnarray}
~~~~~~\rho=\frac{1}{8}I\otimes I\otimes I +\frac{1}{8}(\sum^{3}_{i_{1}=1}t_{i_{1}}^{1}\lambda_{i_{1}}^{(1)}\otimes I\otimes I+\sum^{3}_{i_{2}=1}t_{i_{2}}^{2}I \otimes\lambda_{i_{2}}^{(2)}\otimes I~~~~~~~~~~~~~~~~~~~~~~~~~~~~~~~~~~~~~~~~~~~~~~~~~~~~~~~~~~~~\nonumber\\+\sum^{3}_{i_{3}=1}t_{i_{3}}^{3}I\otimes I\otimes\lambda_{i_{3}}^{(3)})+\frac{1}{8}(\sum^{3}_{i_{1}=1}\sum^{3}_{i_{2}=1}t_{i_{1}i_{2}}^{12}\lambda_{i_{1}}^{(1)}\otimes\lambda_{i_{2}}^{(2)}\otimes I~~~~~~~~~~~~~~~~~~~~~~~~~~~~~~~~~~~~~~~~~~~~~~~~~~~~~~~~~~~~~~~\nonumber\\+\sum^{3}_{i_{1}=1}\sum^{3}_{i_{3}=1}t_{i_{1}i_{3}}^{13}\lambda_{i_{1}}^{(1)}\otimes I\otimes\lambda_{i_{3}}^{(3)}+\sum^{3}_{i_{2}=1}\sum^{3}_{i_{3}=1}t_{i_{2}i_{3}}^{23}I\otimes\lambda_{i_{2}}^{(2)}\otimes \lambda_{i_{3}}^{(3)})~~~~~~~~~~~~~~~~~~~~~~~~~~~~~~~~~~~~~~~~~~~~~~~~~~~~~~~~\nonumber\\
+\frac{1}{8}\sum^{3}_{i_{1}=1}\sum^{3}_{i_{2}=1}\sum^{3}_{i_{3}=1}t_{i_{1}i_{2}i_{3}}^{123}\lambda_{i_{1}}^{(1)}\otimes\lambda_{i_{2}}^{(2)}\otimes\lambda_{i_{3}}^{(3)},~~~~~~~~~~~~~~~~~~~~~~~~~~~~~~~~~~~~~~(1)~~~~~~~~~~~~~~~~~~~~~~~~~~~~~~~~~~~~~~~
\end{eqnarray}
where $\lambda_{i_{j}}^{(j)}$ represents the operator $\lambda_{i_{j}}$ acting on $H_{j}$, $t_{i_{1}}^{1}={\rm tr}(\rho\lambda_{i_{1}}^{(1)}\otimes I\otimes I)$, $\cdots$, $t_{i _{1}i_{2}}^{12}={\rm tr}(\rho\lambda_{i_{1}}^{(1)}\otimes\lambda_{i_{2}}^{(2)}\otimes I)$, $\cdots$,$t_{i_{1}i_{2}i_{3}}^{123}={\rm tr}(\rho\lambda_{i_{1}}^{(1)}\otimes\lambda_{i_{2}}^{(2)}\otimes\lambda_{i_{3}}^{(3)})$. Let $T^{(1)},\cdots$, $T^{(12)}$, $\cdots$, $T^{(123)}$ be the correlation tensors with entries 
$t_{i_{1}}^{1},\cdots,t_{i _{1}i_{2}}^{12},\cdots,t_{i_{1}i_{2}i_{3}}^{123}$, respectively. We equip the tensors with the Hilbert-Schmidt or Frobenius norm $\|\cdot\|$, then $\|T^{(1)}\|^{2}=\sum\limits_{i_{1}=1}^{3}(t_{i_{1}}^{1})^{2}, \cdots$, $\|T^{(12)}\|^{2}=\sum\limits_{i_{1}=1}^{3}\sum\limits_{i_{2}=1}^{3}(t_{i_{1}i_{2}}^{12})^{2},\cdots$, $\|T^{(123)}\|^{2}=\sum\limits_{i_{1}=1}^{3}\sum\limits_{i_{2}=1}^{3}\sum\limits_{i_{3}=1}^{3}(t_{i_{1}i_{2}i_{3}}^{123})^{2}$,
and for a matrix $B\in\mathbb C^{m\times n}$ we also use the trace norm $\|\cdot\|_{tr}$, which is defined as the sum of the singular values,
i.e., $\|B\|_{tr}=\sum\limits_{i}\sigma_{i}={\rm tr}\sqrt{B^{\dag}B}$, where $\sigma_{i}$, $i=1,\cdots, {\rm{min}}(m,n)$, are the singular values of $B$ arranged in descending order. Define that $M(\rho)=\frac{1}{3 }(\|\rho_{1|23}-\rho_{1|23}^{T_{2}}\|_{tr}+\|\rho_{2|13}-\rho_{2|13}^{T_{1}}\|_{tr}+\|\rho_{3|12}-\rho_{3|12}^{T_{1}}\|_{tr})$, where $T_{1}$ is the partial transposition over the first subsystem, $T_{2}$ is the partial transposition over the second  subsystem and $\rho_{i|jk}$ stand  for the bipartite density matrices with respect to subsystem $i$ and $jk$, $j\neq k\neq l\in\{1,2,3\}$.


\begin{lemma}
\label{lemma1}
For any bipartite separable quantum pure state $\rho \in H_{1}^{2}\otimes H_{2}^{2}\otimes H_{3}^{2}$, we denote the bipartitions as follows: $f|gh$, $f\neq g\neq h\in\{1,2,3\}$ and $g<h$. We have $\|\rho_{f|gh}-\rho_{f|gh}^{T_{g}}\|_{tr}\leq \sqrt{3}$.
\end{lemma}

\begin{proof}
For a bipartite separable pure state $|\psi_{\beta}\rangle=|\varphi_{\beta}^{f}\rangle\otimes|\varphi_{\beta}^{gh}\rangle$, using $|\varphi_{\beta}^{f}\rangle\langle\varphi_{\beta}^{f}|=\frac{1}{2}I+\frac{1}{2}\sum\limits_{i_{f}=1}^{3}t_{i_{f}}^{f}\lambda_{i_{f}}^{(f)}$ and $|\varphi_{\beta}^{gh}\rangle\langle\varphi_{\beta}^{gh}|=\frac{1}{4}I\otimes I+\frac{1}{4}(\sum\limits_{i_{g}=1}^{3}t_{i_{g}}^{g}\lambda_{i_{g}}^{(g)}\otimes I+\sum\limits_{i_{h}=1}^{3}t_{i_{h}}^{h}I \otimes\lambda_{i_{h}}^{(h)})+\frac{1}{4}\sum\limits_{i_{g}=1}^{3}\sum\limits_{i_{h}=1}^{3}t_{i_{g}i_{h}}^{gh}\lambda_{i_{g}}^{(g)}\otimes\lambda_{i_{h}}^{(h)}$ , we have\\
\begin{eqnarray}
\||\psi_{\beta}\rangle\langle\psi_{\beta}|-(|\psi_{\beta}\rangle\langle\psi_{\beta}|)^{T_{g}}\|_{tr}~~~~~~~~~~~~~~~~~~~~~~~~~~~~~~~~~~~~~~~~~~~~~~~~~~~~~~~~~~~~~~~~~~~~~
~\nonumber\\=\frac{1}{2}\left[t_{3}^{g}+\sqrt{(t_{31}^{gh})^{2}+(t_{32}^{gh})^{2}+(t_{33}^{gh})^{2}}+\sqrt{(t_{3}^{g}-\sqrt{(t_{31}^{gh})^{2}+(t_{32}^{gh})^{2}+(t_{33}^{gh})^{2}})^{2}}\ \right].~~~~~~~~~~~~
\end{eqnarray}
In Equation (2), if $t_{3}^{g}\geq\sqrt{(t_{31}^{gh})^{2}+(t_{32}^{gh})^{2}+(t_{33}^{gh})^{2}}$, by using $tr[(|\varphi_{\beta}^{f}\rangle\langle\varphi_{\beta}^{f}|)^{2}]=tr[(|\varphi_{\beta}^{gh}\rangle\langle\varphi_{\beta}^{gh}|)^{2}]$, we get $\sum\limits_{i_{g}=1}^{3}f(t_{i_{g}}^{g})^{2}+\sum\limits_{i_{h}=1}^{3}(t_{i_{h}}^{3})^{2}+\sum\limits_{i_{g}=1}^{3}\sum\limits_{i_{h}=1}^{3}(t_{i_{g}i_{h}}^{gh})^{2}=3$, thus $\sum\limits_{i_{g}=1}^{3}(t_{i_{g}}^{g})^{2}\leq 3$, then
\begin{eqnarray}
\||\psi_{\beta}\rangle\langle\psi_{\beta}|-(|\psi_{\beta}\rangle\langle\psi_{\beta}|)^{T_{g}}\|_{tr}=\frac{1}{2}\times(2 t_{3}^{g})\leq\sum\limits_{i_{g}=1}^{3}t_{i_{g}}^{g}\leq\sqrt{3};
\end{eqnarray}
if $t_{3}^{g}<\sqrt{(t_{31}^{gh})^{2}+(t_{32}^{gh})^{2}+(t_{33}^{gh})^{2}}$, by using $\|T^{(ij)}\|^{2}\leq\frac{4(d^{2}-1)}{d^{2}}(i,j=1,2,3)$ in Ref. 21, we have
\begin{eqnarray}
\||\psi_{\beta}\rangle\langle\psi_{\beta}|-(|\psi_{\beta}\rangle\langle\psi_{\beta}|)^{T_{g}}\|_{tr}=\sqrt{(t_{31}^{gh})^{2}+(t_{32}^{gh})^{2}+(t_{33}^{gh})^{2}}\leq \|T^{(gh)}\|\leq\sqrt{3},
\end{eqnarray}
therefore $\||\psi_{\beta}\rangle\langle\psi_{\beta}|-(|\psi_{\beta}\rangle\langle\psi_{\beta}|)^{T_{g}}\|_{tr}\leq \sqrt{3}$.
\end{proof}

The following result follows from Lemma 1.

\begin{theorem}\label{theo1}
For any tripartite qubit quantum state $\rho$ over $H_{1}^{2}\otimes H_{2}^{2}\otimes H_{3}^{2}$, if $\rho$ is bipartite separable, then the inequality $M(\rho)\leq\frac{1}{3}(\sqrt{3}+\sqrt{3}+\sqrt{3})=\sqrt{3}$ holds. Thus if $M(\rho)>\sqrt{3}$, $\rho$ is a genuinely entangled tripartite state. 
\end{theorem}

\begin{proof}
Assume that $\left| {{\psi}} \right\rangle \in H_{1}^{2} \otimes H_{2}^{2}\otimes H_{3}^{2}$ be bipartite separable, for any biseparable mixture state can be written as $\rho=\sum\limits_{i}p_{i}\rho_{i}, \sum\limits_{i}p_{i}=1$. By using the above lemma 1, we take the max  imum value of the trace-norm, we get
\begin{eqnarray}
M(\rho)=\frac{1}{3}(\|\rho_{1|23}-\rho_{1|23}^{T_{2}}\|_{tr}+\|\rho_{2|13}-\rho_{2|13}^{T_{1}}\|_{tr}+\|\rho_{3|12}-\rho_{3|12}^{T_{1}}\|_{tr})~~~~~~~~~~~~~~~~~~~~~~~~~\nonumber\\=
\frac{1}{3}(\|\sum_{i}p_{i}(\rho_{1}\otimes\rho_{23}-(\rho_{1}\otimes\rho_{23})^{T_{2}})\|_{tr}+\|\sum_{i}p_{i}(\rho_{2}\otimes\rho_{13}-(\rho_{2}\otimes\rho_{13})^{T_{1}})\|_{tr}~~\nonumber\\+\|\sum_{i}p_{i}(\rho_{3}\otimes\rho_{12}-(\rho_{3}\otimes\rho_{12})^{T_{1}})\|_{tr})~~~~~~~~~~
~~~~~~~~~~~~~~~~~~~~~~~~~~~~~~~~~~~~~~~~~~\nonumber\\\leq\frac{1}{3}(\sum_{i}p_{i}\|\rho_{1|23}-\rho_{1|23}^{T_{1}}\|_{tr}+\|\rho_{2|13}-\rho_{2|13}^{T_{1}}\|_{tr}+\|\rho_{3|12}-\rho_{3|12}^{T_{1}}\|_{tr})~~~~~~~~~~~~~~~~~~\nonumber\\\leq\frac{1}{3}(\sqrt{3}+\sqrt{3}+\sqrt{3})=\sqrt{3}.~~~~~~~~~~~~~~~~~~~~~~~~~~~~~~~~~~~~~~~~~~~~~~~~~~~~~~~~~~~~~~~~~~~
\end{eqnarray}
\end{proof}

\noindent{\it Example 1:}~Consider the three-qubit Greenberger-Horne-Zeilinger (GHZ) state mixed with white noise, $\rho=\frac{x}{8}I_{8}+(1-x)|{\rm GHZ}\rangle\langle {\rm GHZ}|$, $0\leq x \leq 1$, where $|{\rm GHZ}\rangle=\frac{1}{\sqrt{2}}(|000\rangle+|111\rangle)$, $I_{8}$ the $8\times8$ identity matrix. By calculating, we have\\
$~~~~~~~~~~~M(\rho)=\frac{1}{3 }(\|\rho_{1|23}-\rho_{1|23}^{T_{2}}\|_{tr}+\|\rho_{2|13}-\rho_{2|13}^{T_{1}}\|_{tr}+\|\rho_{3|12}-\rho_{3|12}^{T_{1}}\|_{tr})\\~~~~~~~~~~~~~~~~~~=\frac{1}{3}[3\times(2-2x)]=2-2x$.

Set $f_{1}(x)=M(\rho)-\sqrt{3}$. Theorem 1 says that $\rho$ is a genuine tripartite entangled state when $f_{1}(x)>0$, which happens for $0\leq x<\frac{2-\sqrt{3}}{2}\approx 0.134$.

In Ref. 22, $C_{3}(\rho_{GHZ})=\frac{1}{2}\sqrt{6-25x+\frac{25}{2}x^{2}}$ and set $f_{2}(x)=C_{3}(\rho_{GHZ})-\frac{1}{2}\sqrt{6-25x+\frac{25}{2}x^{2}}$ the lower bound of $f_{2}(x)\geq0$ is used to 
detect genuinely multipartite entangled for $0\leq x < 0.08349$. The comparison is shown in Fig.1, where our result (Theorem 1) is able to detect more genuine entangled states.

\begin{figure}[!htb]
\centerline{\includegraphics[width=1.2\textwidth]{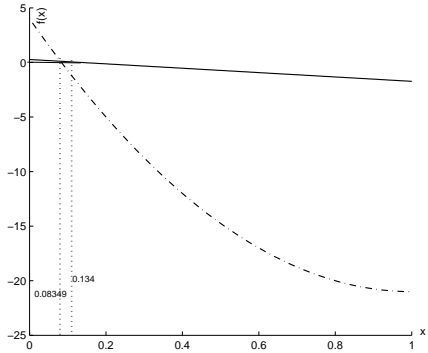}}
\begin{minipage}[t]{16cm}
\caption{Detect genuine entanglement of tripartite qubit quantum states.}
\end{minipage}\par\vglue8pt
\end{figure}

By Theorem 1, when $f_{1}(x)>0$, $\rho$ is a genuine tripartite entangled state for $0\leq x<0.134$
(solid line), while using Theorem 1 in Ref. 22, when $f_{2}(x)\geq0$, $\rho$ is a genuine tripartite entangled state for $0\leq x < 0.08349$(dash-dot line).
\section{Detection of GME for tripartite qudit quantum states}

 In this section we use the norm of correlation tensors to study GME for tripartite qudit quantum states  $\rho$ over $H_{1}^{d}\otimes H_{2}^{d}\otimes H_{3}^{d}$ $(d\geq3)$.

 A pure state $|\psi_{\alpha}\rangle\in H_{1}^{d}\otimes H_{2}^{d}\otimes H_{3}^{d}$ $(d\geq3)$ 
 is said to be biseparable if it can be written as $|\psi_{\alpha}\rangle=|\varphi_{\alpha}^{j}\rangle\otimes|\varphi_{\alpha}^{kl}\rangle$, where $j\neq k\neq l\in\{1,2,3\}$ and $ k< l$. If so, we can write $|\varphi_{\alpha}^{j}\rangle\langle\varphi_{\alpha}^{j}|=\frac{1}{d}I+\frac{1}{2}\sum\limits_{i_{j}=1}^{d^{2}-1}t_{i_{j}}^{j}\lambda_{i_{j}}^{(j)}$and $|\varphi_{\alpha}^{kl}\rangle\langle\varphi_{\alpha}^{kl}|=\frac{1}{d^{2}}I\otimes I+\frac{1}{2d}(\sum\limits_{i_{k}=1}^{d^{2}-1}t_{i_{k}}^{k}\lambda_{i_{k}}^{(k)}\otimes I+\sum\limits_{i_{l}=1}^{d^{2}-1}t_{i_{l}}^{l}I \otimes\lambda_{i_{l}}^{(l)})+\frac{1}{4}\sum\limits_{i_{k}=1}^{d^{2}-1}\sum\limits_{i_{l}=1}^{d^{2}-1}t_{i_{k}i_{l}}^{kl}\lambda_{i_{k}}^{(k)}\otimes\lambda_{i_{l}}^{(l)}$. Thus we can define the constructing matrices with respect to $|\psi_{\alpha}\rangle\langle\psi_{\alpha}|-(|\psi_{\alpha}\rangle\langle\psi_{\alpha}|)^{T_{1}}$.

For the bipartition $1|23$, let\\
\begin{eqnarray}
~~~~~~~~~~~~~~~~~N_{1|23}=\left(\begin{array}{cccccccc}
                                         1 & \frac{1}{2}(T^{(2)})^{t} & \frac{1}{2}(T^{(3)})^{t} & \frac{1}{2}(T^{(23)})^{t} \\
                                         T^{(\tilde{1})} & \frac{1}{2}T^{(\tilde{1}2)} & \frac{1}{2}T^{(\tilde{1}3)} & \frac{1}{2}T^{(\tilde{1}23)}
                                         \end{array}
                                         \right)~~~~~~~~~~~~~\nonumber\\=\left(\begin{array}{c}
                                                   1 \\
                                                   T^{(\tilde{1})}
                                                 \end{array}
                                                 \right)\left(\begin{array}{cccc}
                                                          1 & \frac{1}{2}(T^{(2)})^{t} & \frac{1}{2}(T^{(3)})^{t} & \frac{1}{2}(T^{(23)})^{t}
                                                        \end{array}
                                                        \right),~~~~~~
\end{eqnarray}
where $T^{(\tilde{1})}, T^{(2)}, T^{(3)}, T^{(23)}$ are the column vectors with entries $t_{i_{1}}^{1}:i_{1}=d(d+1)/2$, $\cdots$, $d^{2}-1$, $t_{i_{2}}^{2}$, $t_{i_{3}}^{3}$ and $t_{i_{2}i_{3}}^{23}$, $\|T^{(\tilde{1})}\|^{2}=\sum\limits_{i_{1}=d(d+1)/2}^{d^{2}-1}(t_{i_{1}}^{1})^{2}$, $T^{(\tilde{1}2)}=T^{(\tilde{1})}(T^{(2)})^{t}$, $T^{(\tilde{1}3)}=T^{(\tilde{1})}(T^{(3)})^{t}$, $T^{(\tilde{1}23)}=T^{(\tilde{1})}(T^{(23)})^{t}$, t stands for transpose.

Similarly, for bipartition $2|13$, let
\begin{eqnarray}
G_{2|13}=\left(\begin{array}{ccc}
                                       1 & \frac{1}{2}(T^{(\tilde{1})})^{t} &  \frac{1}{2}(T^{(\tilde{1}3)})^{t} \\
                                       T^{(2)} & \frac{1}{2}T^{(2\tilde{1})} & \frac{1}{2}T^{(2\tilde{1}3)}
                                     \end{array}
                                     \right)~~~~~~~~~~~\nonumber\\=\left(\begin{array}{c}
                                                   1 \\
                                                   T^{(2)}
                                                 \end{array}
                                                 \right)\left(\begin{array}{ccc}
                                                                1 & \frac{1}{2}(T^{(\tilde{1})})^{t} & \frac{1}{2}(T^{(\tilde{1}3)})^{t}
                                                              \end{array}
                                                              \right),~~~
\end{eqnarray}
where $T^{(2\tilde{1})}=T^{(2)}(T^{(\tilde{1})})^{t}$, $T^{(2\tilde{1}3)}=T^{(2)}(T^{(\tilde{1}3)})^{t}$.

For bipartition $3|12$, let
\begin{eqnarray}
S_{3|12}=\left(\begin{array}{ccc}
                                       1 & \frac{1}{2}(T^{(\tilde{1})})^{t} &  \frac{1}{2}(T^{(\tilde{1}2)})^{t} \\
                                       T^{(3)} & \frac{1}{2}T^{(3\tilde{1})} & \frac{1}{2}T^{(3\tilde{1}2)}
                                     \end{array}
                                     \right),~~~~~~~~~
\end{eqnarray}
where $T^{(3\tilde{1})}=T^{(3)}(T^{(\tilde{1})})^{t}$, $T^{(3\tilde{1}2)}=T^{(3)}(T^{(\tilde{1}2)})^{t}$.

By using these constructing matrices and the inequalities for 1-body correlation tensors $\|T^{(j)}\|^{2}\leq\frac{2(d-1)}{d}(j=1,2,3)$ and 2-body correlation tensors $\|T^{(ij)}\|^{2}\leq\frac{4(d^{2}-1)}{d^{2}}(i,j=1,2,3)$, we get the following lemma.

\begin{lemma}
\label{lemma2}
For any bipartite separable quantum pure state $\rho$ over $H_{1}^{d}\otimes H_{2}^{d}\otimes H_{3}^{d}$ $(d\geq3)$, we have\\
(a) If $\rho$ is separable under bipartition $1|23$, then $\|N_{1|23}\|_{tr}\leq \sqrt{\frac{3d^{3}+4d^{2}-7d+2}{2d}}$;\\
(b) If $\rho$ is separable under bipartition $2|13$, then $\|G_{2|13}\|_{tr}\leq \sqrt{\frac{15d^{3}-13d^{2}-4d+4}{2d^{3}}}$;\\
(c) If $\rho$ is separable under bipartition $3|12$, then $\|S_{3|12}\|_{tr}\leq \sqrt{\frac{15d^{3}-13d^{2}-4d+4}{2d^{3}}}$.
\end{lemma}

\begin{proof}
(a) For a bipartite separable pure state $|\psi_{\alpha}\rangle$ under bipartition $1|23$. By Equation (6) it follows that
\begin{eqnarray}
\|N_{1|23}\|_{tr}=\|\left(\begin{array}{c}
                                                   1 \\
                                                   T^{(\tilde{1})}
                                                 \end{array}
                                                 \right)\left(\begin{array}{cccc}
                                                          1 & \frac{1}{2}(T^{(2)})^{t} & \frac{1}{2}(T^{(3)})^{t} & \frac{1}{2}(T^{(23)})^{t}
                                                        \end{array}
                                                        \right)\|_{tr}~~~~~~~~~~~~~~\nonumber\\=\|\left(\begin{array}{c}
                                                   1 \\
                                                   T^{(\tilde{1})}
                                                 \end{array}
                                                 \right)\|\cdot \|\left(\begin{array}{cccc}
                                                          1 & \frac{1}{2}(T^{(2)})^{t} & \frac{1}{2}(T^{(3)})^{t} & \frac{1}{2}(T^{(23)})^{t}
                                                        \end{array}
                                                        \right)\|~~~~~~~~~~~\nonumber\\
=\sqrt{1+\|T^{(\tilde{1})}\|^{2}}\sqrt{1+\frac{1}{4}\|T^{(2)}\|^{2}+\frac{1}{4}\|T^{(3)}\|^{2}+\frac{1}{4}\|T^{(23)}\|^{2}}~~~\nonumber\\\leq \sqrt{1+\|T^{(1)}\|^{2}}\sqrt{1+\frac{1}{4}(2d-\frac{2}{d}-\frac{d}{2}\|T^{(23)}\|^{2})+\frac{1}{4}\|T^{(23)}\|^{2}}\nonumber\\\leq \sqrt{\frac{3d^{3}+4d^{2}-7d+2}{2d}},~~~~~~~~~~~~~~~~~~~~~~~~~~~~~~~~~~~~~~~~~~~~
\end{eqnarray}
where we have used $\||a\rangle\langle b|\|_{tr}=\||a\rangle\|\||b\rangle\|$ for vectors $|a\rangle$ and $|b\rangle$ and $tr(\rho_{1|23}^{2})=tr[(\rho^{1})^{2}]=tr[(\rho^{23})^{2}]=1$.


(b) For a bipartite separable pure state $|\psi_{\beta}\rangle=|\varphi_{\beta}^{2}\rangle\otimes|\varphi_{\beta}^{13}\rangle$,
 Equation (7) implies that\\
$~~~~~~~~~~~~~\|G_{2|13}\|_{tr}=\|\left(\begin{array}{c}
                                                   1 \\
                                                   T^{(2)}
                                                 \end{array}
                                                 \right)\left(\begin{array}{cccc}
                                                         1 & \frac{1}{2}(T^{(\tilde{1})})^{t} & \frac{1}{2}(T^{(\tilde{1}3)})^{t}
                                                        \end{array}
                                                        \right)\|_{tr}\\~~~~~~~~~~~~~~~~~~~~~~~~~=\|\left(\begin{array}{c}
                                                   1 \\
                                                   T^{(2)}
                                                 \end{array}
                                                 \right)\|\cdot\|\left(\begin{array}{ccc}
                                                         1 & \frac{1}{2}(T^{(\tilde{1})})^{t} & \frac{1}{2}(T^{(\tilde{1}3)})^{t}
                                                        \end{array}
                                                        \right)\|\\~~~~~~~~~~~~~~~~~~~~~~~~~=\sqrt{1+\|T^{(2)}\|^{2}}\sqrt{1+\frac{1}{4}\|T^{(\tilde{1})}\|^{2}+\frac{1}{4}\|T^{(\tilde{1}3)}\|^{2}}\\~~~~~~~~~~~~~~~~~~~~~~~~~\leq \sqrt{1+\|T^{(2)}\|^{2}}\sqrt{1+\frac{1}{4}\|T^{(1)}\|^{2}+\frac{1}{4}\|T^{(13)}\|^{2}}$\begin{eqnarray}\leq \sqrt{\frac{15d^{3}-13d^{2}-4d+4}{2d^{3}}}.~~~~~~~~~~~~~~~~~~~~~
\end{eqnarray}


(c) Using the similar method, we obtain $\|S_{3|12}\|_{tr}\leq \sqrt{\frac{15d^{3}-13d^{2}-4d+4}{2d^{3}}}$.
\end{proof}

To detect the genuine tripartite entanglement, we introduce the function $M_{1}(\rho)=\frac{1}{3}(\|N_{1|23}\|_{tr}+\|G_{2|13}\|_{tr}+\|S_{3|12}\|_{tr})$ for any tripartite qudit quantum state $\rho$ over
$H_{1}^{d}\otimes H_{2}^{d}\otimes H_{3}^{d}(d\geq3)$. Then by Lemma $2$,  we set $M_{1}=max\{\sqrt{\frac{3d^{3}+4d^{2}-7d+2}{2d}},\sqrt{\frac{15d^{3}-13d^{2}-4d+4}{2d^{3}}}\}$, then have the following theorem.
\begin{theorem}\label{theo1}
For any tripartite qudit quantum state $\rho$ over $H_{1}^{d}\otimes H_{2}^{d}\otimes H_{3}^{d}$ $(d\geq3)$, if $\rho$ is bipartite separable, then the inequality $M_{1}(\rho)\leq\frac{1}{3}(M_{1}+M_{1}+M_{1})=M_{1}$ holds. Thus if $M_{1}(\rho)>M_{1}$, $\rho$ is genuinely tripartite entangled.
\end{theorem}

\begin{proof}
Assume that $\left| {{\psi_{\alpha}}} \right\rangle \in H_{1}^{d} \otimes H_{2}^{d}\otimes H_{3}^{d}$ be bipartite separable, which will be the following forms:  $|\psi_{\alpha}^{j\mid kl}\rangle=|\varphi_{\alpha}^{j}\rangle\otimes|\varphi_{\alpha}^{kl}\rangle$, where $j\neq k\neq l\in\{1,2,3\}$ and $ k< l$. By using the above lemma 2, we get
\begin{eqnarray}
M(\rho)=\frac{1}{3}(\|N_{1|23}\|_{tr}+\|G_{2|13}\|_{tr}+\|S_{3|12}\|_{tr})~~~~~~~~~~~~~~~~~~~~~~~~~~~~~~~~~~~~~~~\nonumber\\=
\frac{1}{3}(\|\sum\limits_i q_i N_{1|23}(\rho_{1}^{i}\otimes\rho_{23}^{i}-(\rho_{1}^{i}\otimes\rho_{23}^{i})^{T_{1}})\|_{tr}+\|\sum\limits_i r_i G_{2|13}(\rho_{2}^{i}\otimes\rho_{13}^{i}\nonumber\\-(\rho_{2}^{i}\otimes\rho_{13}^{i})^{T_{1}})\|_{tr}+\|\sum\limits_i s_i S_{3|12}(\rho_{3}^{i}\otimes\rho_{12}^{i}-(\rho_{3}^{i}\otimes\rho_{12}^{i})^{T_{1}})\|_{tr}~~~~~~~~~\nonumber\\
\leq\frac{1}{3}(\sum\limits_i q_i\|N_{1|23}\|_{tr}+\sum\limits_i r_i\|G_{2|13}\|_{tr}+\sum\limits_i s_i\|S_{3|12}\|_{tr})~~~~~~~~~~~~~~~~~~\nonumber\\
\leq\frac{1}{3}(M_{1}+M_{1}+M_{1})=M_{1}.~~~~~~~~~~~~~~~~~~~~~~~~~~~~~~~~~~~~~~~~~~~~~~~~~~~~~
\end{eqnarray}
\end{proof}

We consider a special quantum state, if a density matrix is permutational invariant, we can get the following corollary:

\begin{corollary}
If a density matrix is permutational invariant, then $M_{1}(\rho)=\frac{1}{3}(\|N_{1|23}\|_{tr}+\|G_{2|13}\|_{tr}+\|S_{3|12}\|_{tr})\leq\frac{1}{3}(\sqrt{\frac{3d^{3}+4d^{2}-7d+2}{2d}}+2\sqrt{\frac{15d^{3}-13d^{2}-4d+4}{2d^{3}}})$. Thus if $M(\rho)>\frac{1}{3}(\sqrt{\frac{3d^{3}+4d^{2}-7d+2}{2d}}+2\sqrt{\frac{15d^{3}-13d^{2}-4d+4}{2d^{3}}})$, $\rho$ is a genuinely entangled tripartite state.
\end{corollary}

\noindent{\it Example 2:}~Consider the three-qutrit state mixed with white noise,\\ $~~~~~~~~~~~~~~~~~~~~~~~~~~~~\rho=\frac{1-x}{27}I_{27}+x|GHZ\rangle\langle GHZ|$, $0 \leq x \leq 1$, \\where $|GHZ\rangle=\frac{1}{\sqrt{3}}(|000\rangle+|111\rangle+|222\rangle)$, $I_{27}$
the $27\times27$ identity matrix. By calculating, we have\\
$~~~~~~~~~~~~~~~~~~~~~~M_{1}(\rho)=\frac{1}{3}(\|N_{1|23}\|_{tr}+\|G_{2|13}\|_{tr}+\|S_{3|12}\|_{tr})$\\$~~~~~~~~~~~~~~~~~~~~~~~~~~~~~~=\frac{1}{3}(\sqrt{\frac{2}{9}x^{2}+1}+\sqrt{2}x+4x+2)$.\\
By Corollary 2, we set the function \\$~~~~~~~~~~~~~~~~~~~~~f_{1}(x)=M_{1}(\rho)-\frac{1}{3}(\sqrt{\frac{3d^{3}+4d^{2}-7d+2}{2d}}+2\sqrt{\frac{15d^{3}-13d^{2}-4d+4}{2d^{3}}})$, where $d=3$, then $f_{1}(x)>0$ will imply that
$\rho$ is GME, which happens when $0.708 < x \leq 1$. 

On the other hand, according to Corollary 2 of Ref. 16, another function \\$~f_{2}(x)=C_{n-1}(\rho_{2-sep})- {\rm max}[d^{n}-2d^{\frac{n}{2}}+1,(d-1)(d^{n-1}-1)+d^{n-1}-1-\frac{n}{n-1}(d^{n-3}-1)]$, \\was used to detect
GME and 
their criterion certifies this when 
$x > 0.89443$. The comparison is given in See Fig.2, which clearly shows that our Corollary 2 is stronger than the criterion given in Ref. 16.

\begin{figure}[!htb]
\centerline{\includegraphics[width=1.4\textwidth]{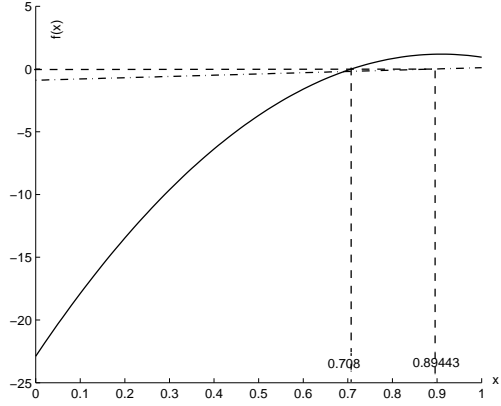}}
\begin{minipage}[t]{16cm}
\caption{Detect genuine entanglement of tripartite qudit quantum states.}
\end{minipage}\par\vglue8pt
\end{figure}
By Corollary 2, when $f_{1}(x)>0$, $\rho$ is a genuine tripartite entangled state for $0.708 < x \leq 1$(solid line), while using Corollary 2 of Ref. 16, when $f_{2}(x)>0$, $\rho$ is a genuine tripartite entangled state for $x > 0.89443$(dash-dot line).

Moreover, Ref. 12 used the lower bound of the following formular:\\$~~~~~~~~~~~~~~~~~~~~~~C_{N}(\rho)>2^{1-\frac{N}{2}}\sqrt{2^{N}-4+\frac{2}{d}-2\sum\limits_{k=1}^{\frac{N-1}{2}}\frac{N!
}{k!(N-k)!d^{k}}} $ $(N=3,d=3)$, \\to detect GME for tripartite $\rho$ at $x > 0.83485$.

 Also Theorem 5 in Ref. 23 considered the overlap between the maximal quantum mean value and the classical bound of the Clauser-Horne-Shimony-Holt (CHSH) inequality and their result was that
tripartite $\rho$ is GME when $x > 0.731621$.

 Therefore, for $d = 3$, our Corollary 2 detects more genuine tripartite entangled states than those of Refs. 12, 16 and 23. However, for $d = 4$, our result is not superior than that of Ref. 23.

\section{Conclusion}
We have studied the genuine multipartite entanglement for general tripartite systems based on partial transposition and the norms of
correlation tensors. We obtain sufficient conditions of the genuine entanglement for general tripartite qubit quantum states and have
derived the GME criterion under bipartition by constructing a matrix for tripartite qudit quantum states. Using detailed examples we have shown that our criteria detect more genuine entanglement than previous studies. Genuine multipartite entanglement plays a significant role in many quantum information processing. Our approach and results may be helpful for the further research on genuine multipartite entanglement.

\section*{Acknowledgments}

This work is supported by the National Natural Science Foundation of China under Grant Nos. 11101017, 11531004, 11726016, 11675113 and 11772007, and Simons Foundation under Grant No. 523868 and Beijing Natural Science Foundation under Grant No. Z180005.

\end{document}